\documentclass[11pt]{article}
\usepackage[totalwidth=13.0cm,totalheight=20.0cm]{geometry}
\usepackage{latexsym,amsthm,amsmath,amssymb,url}
\usepackage[ruled, linesnumbered]{algorithm2e}

\usepackage{tikz,authblk}
\usetikzlibrary{decorations.pathreplacing}

\newtheorem{lemma}{Lemma}

\newtheorem{theorem}{Theorem}

\newtheorem{claim}{Claim}

\newcommand{\2}{\vspace{0.2cm}}

\begin{document}

\title{Chinese Postman Problem on Edge-Colored Multigraphs}

\author[1]{Gregory Gutin}\author[1]{Mark Jones}\author[1]{Bin Sheng}\author[1]{Magnus Wahlstr{\"o}m}
\author[2,3]{Anders Yeo}
\affil[1]{Department of Computer Science, Royal Holloway, University of London, TW20 0EX, Egham, Surrey, UK}
\affil[2]{Engineering Systems and Design, Singapore University of Technology and Design, 8 Somapah Road 487372, Singapore} 
\affil[3]{Department of Mathematics, University of Johannesburg, Auckland Park, 2006 South Africa}
\maketitle

\maketitle

\pagestyle{plain}
\begin{abstract} It is well-known  that the Chinese Postman Problem on undirected and directed graphs is polynomial-time solvable. We extend this result to 
edge-colored multigraphs. Our result is in sharp contrast to the Chinese Postman Problem on mixed graphs, i.e., graphs with directed and undirected edges, for which the problem is NP-hard.
\end{abstract}

\section{Introduction}\label{sec1}

In this paper, we consider edge-colored multigraphs. In such multigraphs, each edge is assigned  {a} color; a multigraph $G$ is called {\em $k$-edge-colored} if only  {colors from $[k]:=\{1,2,\dots ,k\}$} are used in $G$.
A walk\footnote{ {Terminology on walks used in this paper is given in the next section.}} $W$ in an edge-colored multigraph is called {\em properly colored (PC)} if no two consecutive edges of $W$ have the same color.
PC walks are of interest  in graph theory applications, e.g., in genetic and molecular biology \cite{Pe,Sz+,Sz++}, in design of printed circuit and wiring boards \cite{TsChLe}, and in channel assignment in wireless networks \cite{Ah,Sa+}. They are also of interest in graph theory itself as generalizations of walks in undirected and directed graphs. Indeed, if we assign different colors to all edges of an undirected multigraph, every walk not traversing the same edge twice becomes PC. Also, consider the standard transformation from a directed graph $D$ into a 2-edge-colored graph $G$ by replacing every arc $uv$ of $D$ by a path with  {a} blue edge $uw_{uv}$ and  {a} red edge $w_{uv}v$, where $w_{uv}$ is a new vertex \cite{JBJGG}.  Clearly, every directed walk in $D$  corresponds to a PC walk in $G$ (with end-vertices in $V(G)$) and vice versa. 
There is an extensive literature on PC walks: for a detailed survey of pre-2009 publications, see Chapter 16 of \cite{JBJGG}, more recent papers include 
\cite{ADFKMMS,FuMa,Lo1,Lo2,Lo3}.

A walk is {\em closed} if it starts and ends in the same vertex.  (A closed walk $W$ has no last edge, every edge in $W$  has a following edge; if $W$ is PC, each edge of $W$ is of different color to the following edge.)
An {\em Euler trail} in a multigraph $G$ is a closed walk which traverses each edge of $G$ exactly once.
PC Euler trails  were one of the first types of PC walks studied in the literature and the first papers that studied PC Euler trails were motivated by theoretical questions \cite{FSW,Kotz} as well as questions in molecular biology \cite{Pe}. 
To formulate a characterization of edge-colored graphs with PC Euler trails by Kotzig \cite{Kotz}, we introduce additional terminology.
A vertex in an edge-colored multigraph is {\em balanced} if no color appears on more than half of the edges incident with the vertex, and \emph{even} if it is of even degree. We say that an edge-colored graph is \emph{PC Euler} if it contains a PC Euler trail.

\begin{theorem} \label{thm:kotzig}\cite{Kotz} 
An edge-colored multigraph $G$ is PC Euler if and only if $G$ is connected and every vertex of $G$ is balanced and even. 
\end{theorem} 

Benkouar {\em et al.} \cite{BMPS} described a polynomial-time algorithm to find a PC Euler trail in an edge-colored multigraph, if it contains one.
Studying DNA physical mapping, Pevzner \cite{Pev} came up with a simpler  polynomial-time algorithm solving the same problem. 


In this paper, we consider the {\em Chinese Postman Problem on edge-colored graphs (CPP-ECG)}: {\em given a connected edge-colored multigraph $G$ with non-negative weights on  {its} edges, find a PC closed walk in $G$ which traverses all edges of $G$ and has the minimum weight\footnote{The weight of a walk is the sum of the weights of its edges.} among such walks}. 

 Observe that to solve CPP-ECG, it is enough to find a PC Euler edge-colored multigraph $G^*$ of minimum weight such that $V(G^*)=V(G)$ and for every pair of distinct vertices $u,v$  and color $i$, $G^*$ has $p^*>0$ parallel edges between vertices $u$ and $v$ of color $i$ if and only if $G$ has  {at least one and} at most $p^*$ edges of color $i$ between $u$ and $v$. (To find the actual walk, we can use the algorithm from \cite{BMPS} or \cite{Pe}.)

 {CPP-ECG is 
a generalization of the PC Euler trail problem as an instance $G$ has a PC Euler trail if and only if $G^*=G$.}
CPP-ECG is also a generalization of the Chinese Postman Problem (CPP) on both undirected and directed multigraphs (the arguments are the same as for  PC walks above). 
  { 
 However, while CPP on both undirected and directed multigraphs has a solution on every connected multigraph $G$, it is not the case for CPP-ECG. Indeed, there is no solution on any connected edge-colored multigraph containing a vertex incident to edges of only one color.
}

{It is not hard to solve CPP on undirected and directed multigraphs \cite{KoVy}. For a directed multigraph $G$, we construct a flow network $N$ by assigning lower bound 1, upper bound $\infty$ and cost $\omega(uv)$ to each arc $uv$, where $\omega(uv)$ is the weight of $uv$ in $G$. A minimum-cost circulation in $N$ viewed as an Euler directed multigraph corresponds to a CPP solution and vice versa.  For an undirected multigraph $G$, we construct an edge-weighted complete graph $H$ whose vertices are odd degree vertices of $G$ and the weight of an edge $xy$ in $H$ equals the minimum weight path between $x$ and $y$ in $G$. Now find a minimum-weight perfect matching $M$ in $H$ and add to $G$ a minimum-weight path of $G$ between $x$ and $y$ for each edge $xy$ of $M$. The resulting Euler multigraph corresponds to a CPP solution and vice versa.}

{We will prove that  CPP-ECG is polynomial-time solvable as well.  Note that our proof is significantly more complicated than that for CPP on undirected and directed graphs. As in the undirected case, we construct an auxiliary edge-weighted complete graph $H$ and seek a minimum-weight perfect matching $M$ in it. However, the construction of $H$ and the arguments justifying the appropriate use of $M$ are significantly more complicated. This can partially be explained by the fact that CPP-ECG has no solution on many edge-colored multigraphs}.

Note that there is another generalization of CPP on both undirected and directed multigraphs, namely, CPP on mixed multigraphs, i.e., multigraphs that may have both edges and arcs. However, CPP on mixed multigraphs is NP-hard \cite{Papad}. {It is fixed-parameter tractable when parameterized by both number of edges and arcs \cite{vanBevern,GJS} and W[1]-hard when parameterized by pathwidth \cite{GJW}.}
For more information on the classical and parameterized complexity of  CPP and its generalizations, see an excellent survey by van Bevern {\em et al.} \cite{vanBevern}.


\section{Preliminaries}\label{sec:pre}

 {{\bf Walks.}
A {\em walk} in a multigraph is a sequence $W=v_1e_1v_2\dots v_{p-1}e_{p-1}v_p$ of alternating vertices and edges such that vertices $v_i$ and $v_{i+1}$ are end-vertices of edge $e_i$ for every $i\in [p-1]$.  
A walk $W$ is {\em closed} ({\em open}, respectively) if $v_1=v_p$ ( $v_1\neq v_p$, respectively). 
A {\em trail} is a walk in which all edges are distinct. 
}

 {
For technical reasons we will consider walks with fixed end vertices and call them {\em fixed end-vertex (FEV) walks}. Note that an open walk is necessarily  an FEV walk since the end-vertices are predetermined, whereas any vertex in a closed walk can be viewed as its two end-vertices and thus fixing such a vertex is somewhat similar to assigning a root vertex in a tree.
An FEV walk $W=v_1e_1v_2\dots v_{p-1}e_{p-1}v_p$ is {\em PC} in an edge-colored graph  if the colors of $e_i$ and $e_{i+1}$ are different for every $i\in [p-2]$. Note that we do not require that colors of $e_{p-1}$ and $e_1$ are different even if $v_1=v_p$. Thus, a PC FEV walk might not be a PC walk if $v_1=v_p$.
}

 Let   {$e=xy$} be an edge in an edge-colored multigraph $G$. The operation of {\em double subdivision} of   {$e$} replaces $e$ with an $(x,y)$-path $P_e$ with three edges { such that the weight of $P_e$ equals}    that of edge $e$.
 
{\em It is easy to see that in our study of PC walks, we may restrict ourselves to graphs rather than multigraphs. Indeed, it suffices to double subdivide every parallel edge $e$ and assign the original color of $e$ to the first and third edges of $P_e$ and a new color to the middle edge.} \\

\noindent {\bf Finding PC FEV walks.}
Let $\mathbb{R_+}$ denote the set of non-negative real numbers.
To show a polynomial-time algorithm for CPP-ECG, we will use the following:

\begin{lemma}
  \label{lemma:mincostwalk}
  Let $G=(V,E)$ be a  {$k$}-edge-colored graph and $\omega: E \to \mathbb{R_+}$ a weight function. Let vertices $u, v \in V$ and edge colors $c_1, c_2$ be given, where we may have $u=v$. In polynomial time we can find a minimum-weight PC FEV walk from $u$ to $v$ in $G$ whose first edge has color $c_1$ and whose last edge has color $c_2$, or conclude that there is no such PC {FEV} walk in $G.$
\end{lemma}
\begin{proof}

  Define an auxiliary digraph $H$ as follows. Let  the vertex set of $H$ be  {$\{(u,0)\}\cup \{(x,i): x \in V, i \in [k]\}$}. For every edge $xy \in E$, of color $i$, we add to $H$ all arcs $(x,j)(y,i)$ and $(y,j)(x,i)$ where $j \in [k]$, $j \neq i$. We also add   { an arc} from  {$(u,0)$} to $(z,c_1)$ for every edge $uz \in E$ of color $c_1$. Every arc in $H$ retains the weight of the corresponding edge in $G$. 
  We claim that the minimum-weight {PC}  {FEV }walk we seek in $G$ corresponds to a minimum-weight directed path from   {$(u,0)$} to $(v,c_2)$ in $H$, which can be found 
{ in polynomial time,}
e.g., using Dijkstra's algorithm.

On the one hand,  
 {
let $(x_1,d_1)(x_2,d_2) \dots (x_{\ell},d_{\ell})$ be a directed path in $H$ such that $(x_1,d_1)=(u,0)$ and $(x_{\ell},d_{\ell}) = (v,c_2)$. Then by construction, $x_1e_2x_2 \dots e_{\ell}x_{\ell}$, where $e_i$ is an edge between $x_{i-1}$ and $x_i$ of color $d_i$, is a PC FEV walk in $G$ with required properties. On the other hand, consider a minimum-weight PC FEV walk $W$ in $G$ with the properties requested. Orient the edges of the walk away from $u$. We may assume that no vertex has two in-coming directed edges in the walk of the same color, as the walk could otherwise be shortened.  As above it is not hard to verify that the walk corresponds to a directed path $P$ in $H$ from $(u,0)$ 
to $(v,c_2)$. By construction, the weight of $P$ equals that of $W$. It remains to observe that $P$ is a minimum-weight directed path from $(u,0)$ to $(v,c_2)$, as otherwise there is a PC FEV walk between $u$ and $v$ with required edge colors of weight smaller than $W$, a contradiction. 
}



Finally, we observe that the construction works without modification if $u=v$. 
\end{proof}

\section{Main Result}\label{sec:proof}

We are now ready to prove the main result.

\begin{theorem} \label{main}
We can solve CPP-ECG in polynomial time.
\end{theorem}

\begin{proof}
 {Let $G$ be an input of CPP-ECG. That is, $G$ is a connected $k$-edge-colored graph with at least one edge.}
We may assume that no vertex of $G$ is   { incident with edges of a single color only}
($G$ cannot have PC closed walks  {through such vertices}).
We may also assume that $k\ge 2$
is odd (if not, we double subdivide an edge $e$ of $G$ and assign a new color to the middle edge of $P_e$ and the original color of $e$ to the other two edges). 

For a vertex $u \in V(G)$ and color $i \in [k]$, let $d_i(u)$ be the number of edges incident with $u$ of color $i$.
Let $d(u)$ ($=\sum_{i=1}^k d_i(u)$) be the degree of $u$ in $G$. We say that color $i$ is \emph{dominant for $u$ in $G$}
if $2d_i(u)>d(u)$; note that a vertex has at most one dominant color, and is balanced if and only if it has no dominant color.

\begin{figure}\centering
\begin{tikzpicture}
\draw [dotted](5,7)[fill]circle [radius=0.07]node [left]{$u_4$}--(7,7)[fill]circle [radius=0.07] node [right]{$u_3$};
\draw [dashed, ultra thick](5,9)circle [radius=0.07] node [left]{$u_1$}--(7,9)[fill]circle [radius=0.07]node [right]{$u_2$};
\draw [dashed, ultra thick](5,7)--(7,7);
\draw [](5,9)--(7,7);
\draw [](5,7)--(5,9);
\draw [](7,9)--(7,7);
\draw [dotted](5,7)--(7,9);
\draw [](11.9,7)node[left]{color 1}  --(12.9,7) ;
\draw [dotted](11.9,9)node[left]{color 3}  --(12.9,9) ;
\draw [dashed, ultra thick](11.9,8)node[left]{color 2}  --(12.9,8) ;

\end{tikzpicture}
\caption{3-edge-colored graph $G$}
\label{fig:ExampleG}
\end{figure}
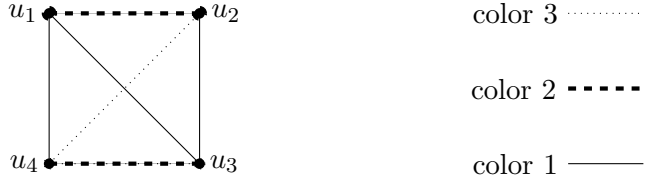

\begin{figure}
\centering
\begin{tikzpicture}

\draw [](4.5,0)node[left]{$X(u_4)$}--(4.5,0);
\draw [](7.5,0)node[left]{$X_3(u_3)$}--(7.5,0);
\draw [](8,2)node[right]{$X_2(u_3)$}--(8,2);

\draw [](5,1)[fill]circle[radius=0.07];
\draw [](5,0)[fill]circle[radius=0.07];
\draw [](5,-1)[fill]circle[radius=0.07];
\draw [](5,4)[fill]circle[radius=0.07];
\draw [](5,3)[fill]circle[radius=0.07]node[below]{$X_2(u_1)$};
\draw [](5,5)[fill]circle[radius=0.07];
\draw [](5,6)[fill]circle[radius=0.07]node[above]{$X_3(u_1)$};
\draw [](8,-1)[fill]circle[radius=0.07];
\draw [](8,0)[fill]circle[radius=0.07];
\draw [](8,1)[fill]circle[radius=0.07];
\draw [](8,2)[fill]circle[radius=0.07];
\draw [](8,4)[fill]circle[radius=0.07];
\draw [](8,5)[fill]circle[radius=0.07];
\draw [](8,6)[fill]circle[radius=0.07]node[above]{$X(u_2)$};
\draw [](10,-1)[fill]circle[radius=0.07];
\draw [](11,0)[fill]circle[radius=0.07];
\draw [](10,1)[fill]circle[radius=0.07];

\draw [dashed](5, 0) ellipse (.5 and 1.2);
\draw [dashed](8, 0.5) ellipse (.5 and 1.7);
\draw [dashed](5, 4.5) ellipse (.5 and 1.7);
\draw [dashed](8, 5) ellipse (.5 and 1.2);
\draw [dashed](2.8, 5) ellipse (0.9 and 1.5);
\draw [dashed](10.2, 0) ellipse (0.9 and 1.5);

\draw [](8,1)--(10,-1);
\draw [](8,1)--(11,0);
\draw [](8,1)--(10,1);
\draw [](8,2)--(10,-1);
\draw [](8,2)--(11,0);
\draw [](8,2)--(10,1);
\draw [](8,-1)--(10,-1);
\draw [](8,-1)--(11,0);
\draw [](8,-1)--(10,1);
\draw [](8,0)--(10,-1);
\draw [](8,0)--(11,0);
\draw [](8,0)--(10,1);
\draw [](2,5)--(5,3);
\draw [](2,5)--(5,4);
\draw [](2,5)--(5,5);
\draw [](2,5)--(5,6);
\draw [](3,6)--(5,3);
\draw [](3,6)--(5,4);
\draw [](3,6)--(5,5);
\draw [](3,6)--(5,6);
\draw [](3,4)--(5,3);
\draw [](3,4)--(5,4);
\draw [](3,4)--(5,5);
\draw [](3,4)--(5,6);
\draw [](2,5)[fill]circle[radius=0.07]node[left]{$Y(u_1)$}--(3,4)[fill]circle[radius=0.07];
\draw [](3,4)[fill]circle[radius=0.07]--(3,6)[fill]circle[radius=0.07];
\draw [](3,6)[fill]circle[radius=0.07]--(2,5)[fill]circle[radius=0.07];
\draw [](10,-1)[fill]circle[radius=0.07]--(10,1);
\draw [](11,0)[fill]circle[radius=0.07]--(10,-1);
\draw [](10,1)[fill]circle[radius=0.07]--(11,0)node[right]{$Y(u_3)$};

\draw [dotted](5,3)--(8,6);
\draw [dotted](5,3)--(8,5);
\draw [dotted](5,3)--(8,4);

\draw [dotted](5,1)--(8,6);
\draw [dotted](5,1)--(8,5);
\draw [dotted](5,1)--(8,4);

\draw [dotted](5,0)--(8,6);
\draw [dotted](5,0)--(8,5);
\draw [dotted](5,0)--(8,4);

\draw [dotted](5,-1)--(8,6);
\draw [dotted](5,-1)--(8,5);
\draw [dotted](5,-1)--(8,4);

\draw [dotted](8,2)--(5,-1);
\draw [dotted](8,2)--(5,0);
\draw [dotted](8,2)--(5,1);

\draw [dotted](5,3)--(8,2);
\draw [dotted](8,2)--(8,4);
\draw [dotted](5,1)--(5,3);
\draw [dotted](5,0) to [out=180,in=-170] (5,3);
\draw [dotted](5,-1) to [out=180,in=-170] (5,3);
\draw [dotted](8,2) to [out=0,in=-10] (8,5);
\draw [dotted](8,2) to [out=0,in=-10] (8,6);
\draw [dotted](5,-1)--(8,4);
\draw [dotted](5,0)--(8,5);
\draw [dotted](5,1)--(8,6);

\draw [dashed](7.5,1.5)--(8.5,1.5);
\draw [dashed](4.5,3.5)--(5.5,3.5);
\draw [dashed](7.5,4.5)--(8.5,4.5);
\draw [dashed](7.5,5.5)--(8.5,5.5);
\draw [dashed](4.5,0.5)--(5.5,0.5);
\draw [dashed](4.5,-0.5)--(5.5,-0.5);

\draw [](11.9,6)node[left]{artificial edge}  --(12.9,6) ;
\draw [dotted](11.9,5)node[left]{non-artificial edge}  --(12.9,5) ;

\end{tikzpicture}
\caption{Constructed graph $H$ from the graph of Figure~\ref{fig:ExampleG}.
  Some edges, including artificial edges within $X(u_2)$ and $X(u_4)$, are omitted for clarity. }
 \label{fig:ExampleH}
\end{figure}
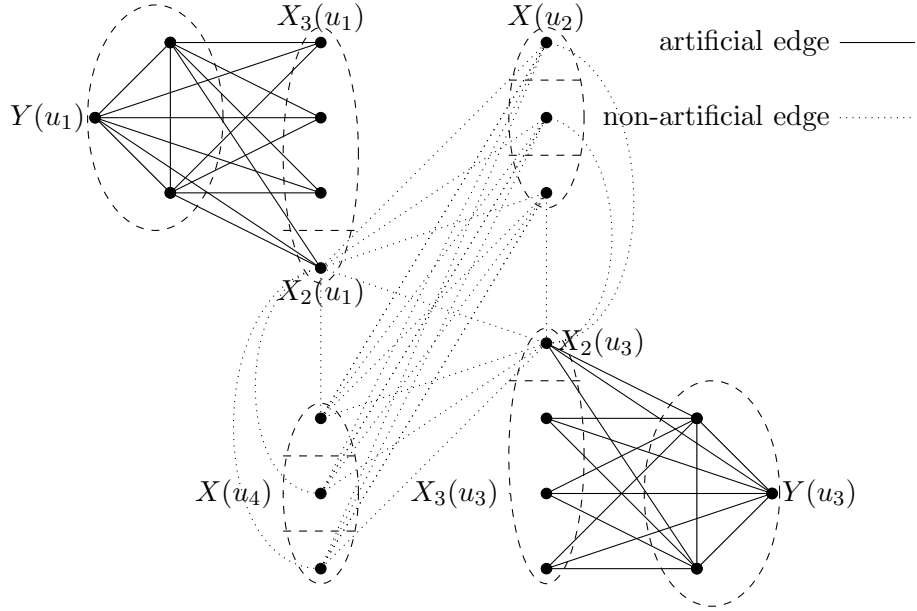

 {We now show how to construct, in polynomial time, an undirected graph $H$ such that $H$ has a perfect matching if and only if $G$ has a PC closed walk traversing all edges of $G$, and the minimum weight of a perfect matching in $H$ is equal to the minimum weight of such a walk in $G$ minus the weight of $G$. As computing a minimum-weight perfect matching can be done in polynomial time, the claim follows.}
An example is shown in Figures~\ref{fig:ExampleG} and~\ref{fig:ExampleH}.

We will build the undirected graph $H$ as follows.
Define $\theta_i(u)$ as follows:
\[
 \theta_i(u) = \max \{0, d(u) - 2d_i(u) \}.
\]

Let $X_i(u)$ be a set of independent vertices of size $\theta_i(u)$ and let $X(u) = \bigcup_{i=1}^k X_i(u)$. 
We now consider the cases when $u$ is balanced and when $u$ is not balanced, separately.

\2

{\em Case 1: $u$ is not balanced.}
Let $Y(u)$ be a set of $(k-2)d(u)$ independent vertices and add all possible edges between $Y(u)$ and $X(u)$ and all possible edges within $Y(u)$.
Let the weight of all these edges be zero and let us call them {\em artificial} edges. Let $Z(u)=X(u) \cup Y(u)$.

\2

{\em Case 2: $u$ is balanced.} Add all possible edges within $X(u)$.
Let the weight of all these edges be zero and let call them {\em artificial} edges. Let $Z(u)=X(u)$.

\2

For every pair  $a \in X_i(u)$ and $b \in X_j(v)$ of distinct vertices such that $ab$ is not an artificial edge, $i,j \in [k]$ and $u,v \in V(G)$ 
(we may have $i=j$ and/or $u=v$), add an 
edge between $a$ and $b$ with the weight equal to  {that} of a minimum-weight PC {FEV} walk from $u$ to $v$ in $G$, starting in color $i$ and ending in color $j$  {if one exists, otherwise add no edge $ab$}; this can be computed in polynomial time by Lemma~\ref{lemma:mincostwalk}.
This completes the description of $H$.

Assume that $H$ has a perfect matching and $M$ is a minimum-weight perfect matching in $H$. We will show that the weight of $M$ plus the weight of all edges in $G$ is the weight of 
an optimal solution to the CPP-ECG instance. 

We begin with an observation about the structure of $M$. 

\begin{claim}
  \label{claim:M-props}
  We say that a vertex $a \in X(u)$ for some $u \in   {V(G)}$ is affected by $M$ if 
  $a$ is incident with a non-artificial edge in $M$.
 {Then $M$ has the following properties.}
  \begin{enumerate}
  \item If $u$ is unbalanced with dominant color $i \in [k]$,
    then at least $2d_i(u)-d(u)$ vertices of $X(u)$ are affected by $M$.
    Furthermore $X_i(u)=\emptyset$. 
  \item If $d(u)$ is odd, then an odd number of vertices of $X(u)$ are affected by $M$.
  \item If $d(u)$ is even, then an even number of vertices of $X(u)$ are affected by $M$. 
  \end{enumerate}
  Furthermore,  {for any matching $M_0$ in $H$ with no artificial edges that has the above properties, $M_0$} can be completed to a perfect matching by adding
  artificial edges. 
\end{claim}
\begin{proof}
  \emph{1.} Assume that $u$ is unbalanced with dominant color $i$, i.e., $2d_i(u)>d(u)$, $i \in [k]$. Then necessarily, $2d_j(u)<d(u)$ for every other color $j \in [k]$, 
  hence $\theta_j(u)=d(u)-2d_j(u)$ if $i \neq j$, whereas $\theta_i(u)=0$. Thus 
\[
|X(u)|=\sum_{j \neq i} (d(u)-2d_j(u))= (k-2)d(u) + (2d_i(u)-d(u)),
\]
where the last equality uses~$\sum_{j \neq i} d(u)=(k-1)d(u)$ and~$\sum_{j \neq i} d_j(u)=d(u)-d_i(u)$.
  Artificial edges on $X(u)$ can only match vertices of $X(u)$ against $Y(u)$. Since $|Y(u)|=(k-2)d(u) < |X(u)|$, this leaves  {at least} $|X(u)|-|Y(u)|=2d_i(u)-d(u)$
  vertices in $X(u)$ which must be affected by $M$. Finally, $X_i(u)=\emptyset$ since $\theta_i(u)=0$.

  \emph{2 and 3.} Assume first that $u$ is unbalanced, so there is a set of vertices $Y(u)$. Since $k$ is odd, the parity of $|Y(u)|$ matches the parity of $d(u)$,
  hence an odd (resp.~even) number of vertices of $X(u)$ are matched against $Y(u)$ if and only if $d(u)$ is odd (resp.~even). 
  Now, as calculated in the previous paragraph, $|X(u)|=(k-3)d(u)+2d_i(u)$, which is always even. Furthermore, the affected vertices of $X(u)$
  are exactly those not matched against $Y(u)$. The claim follows.

  Finally, if $u$ is balanced, then $|X(u)|=\sum_j (d(u)-2d_j(u))=(k-2)d(u)$, which again has the same parity as $d(u)$. 
  Since every artificial edge matches two vertices of $X(u)$, and the
  remaining vertices are exactly the affected vertices in $X(u)$, the
  claim follows.

  \emph{Completing a non-perfect matching.} Let $M_0$ be a matching in $H$
  such that for every vertex $u$, (1) if $u$ has a
  dominant color $i$, then at least $2d_i(u)-d(u)$ vertices of $X(u)$
  are affected by $M_0$, and (2) an odd number of vertices of $X(u)$
  are affected by $M_0$ if and only if $d(u)$ is odd. By the above,
  the second point here implies that the number of unmatched vertices
  of $Z(u)$ is even for every vertex $u$. If $u$ is balanced, then 
  $Z(u)=X(u)$ is a clique and we can add artificial edges
  from the clique. If $u$ is unbalanced, then $Y(u)$ is entirely
  unmatched, and by the first point here, the number of unmatched
  vertices in $X(u)$ is at most $|Y(u)|$. Hence the completion is
  possible. 
\end{proof}

 { 
Let $u$ and $v$ be vertices of $G$ and let $i$ and $j$ be colors. Let $e=ab$ be an arbitrary non-artificial edge in $H$, with
$a \in X_i(u)$ and $b \in X_j(v)$, where we may have $u=v$ and/or $i=j$.   An \emph{$e$-walk}
is a PC FEV walk in $G$, starting at $u$ with an edge of color $i$, and ending at $v$ 
with an edge of color $j$.
}

We will show the theorem in two parts. First, we will show that for any perfect matching $M'$ of $H$, 
if we add an $e$-walk to $G$ for every non-artificial edge $e \in M'$, then the resulting multigraph $G'$    { is PC Euler}. Here, to \emph{add} an $e$-walk to $G$ means to duplicate every 
edge along the walk, duplicating an edge multiple times if it occurs in the walk multiple times.
Second, we show that for any PC Euler graph $G'=(V, E \cup W)$, obtained by duplicating edges in $E$,
there exists a perfect matching $M'$ of $H$ such that $W$ can be decomposed into  { a set $\mathcal{F}$ of $e$-walks, where there is an $e$-walk in $\mathcal{F}$ if and only if $e$ is a non-artificial edge in $M'$}. This will settle the result. 

We need an observation about the effect of adding an $e$-walk to a graph.

\begin{claim} \label{claim:e-walks}
  Let $e \in H$ be a non-artificial edge. Adding an $e$-walk $F$ to $G$ has the following effects.
  \begin{enumerate}
  \item For any vertex $u$, the parity of $d(u)$ changes if and only if $F$ 
    is open  and $u$ is either its first or last vertex.
  \item If $u \in V(G)$ is neither the first nor the last vertex of the walk, then for every $i \in [k]$,
    the value of $d(u)-2d_i(u)$ is non-decreasing in the process.
  \item If $F$ is closed, let $u$ be its end-vertex, and let $i$ ($j$, respectively) be the colors 
    of its first (last, respectively) edges. Then for any $c \in [k]$, the value of $d(u)-2d_c(u)$
    increases by at least 2 if $c \notin \{i,j\}$; is non-increasing if $c \in \{i,j\}$ 
    and $i \neq j$; and decreases by at most 2 if $i=j=c$.
  \item If $F$ is open, let $u$ be an end-vertex of $F$, without loss of generality, the first one. 
    Let $i$ be the color of the first edge. Then for any $j \in [k]$,
    the value of $d(u)-2d_j(u)$ decreases by at most one if $j=i$, and increases by 
    at least one, otherwise.
  \end{enumerate}
\end{claim}
\begin{proof}
  The first item is easy. For the second item, we just observe that a single
  transition through $u$ increases $d(u)$  {by} 2 and $d_i(u)$ by at most 1. 
  Since the graph has no loops, the local effect on $u$ of duplicating $F$
  decomposes into transitions, hence the second item holds. By the same argument, 
  if $u$ is an endpoint of $F$, then all visits to $u$ except the first and/or last one
  decompose into transitions. This leaves only the first and last edges of $F$,
  and their effects on the end-vertices of $F$, to consider. The claims in items~3 and~4
  follow by considering all possibilities for these two edges.
\end{proof}

We are now ready to show the first part of the theorem, as announced above. Let $M'$ be an arbitrary
perfect matching in $H$, and let $G'$ be the result of adding an arbitrary $e$-walk,  {which has the same weight as edge $e$,} to $G$ for every
non-artificial edge $e \in M'$. We will show that $G'$ is PC Euler.
By Theorem~\ref{thm:kotzig}, we need to show three conditions: $G'$ is connected, every
vertex in $G'$ is even, and every vertex in $G'$ is balanced. The first condition follows since $G$ is connected; 
the second condition follows from Claim~\ref{claim:M-props}(2--3) and Claim~\ref{claim:e-walks}(1). 
It remains to show that every vertex is balanced in $G'$, i.e., for every $u \in V(G)$
and every $i \in [k]$, it holds in $G'$ that $d(u) \geq 2d_i(u)$. 
We break this down into two cases.

\emph{Case 1: $i$ is the dominant color for $u$ in $G$.}
In this case, $d(u)-2d_i(u)<0$ in $G$, and we need to show that this value is nonnegative in $G'$. 
By Claim~\ref{claim:M-props}(1), at least $2d_i(u)-d(u)$ vertices of $X(u)$
are affected by $M'$, and since $X_i(u)=\emptyset$, Claim~\ref{claim:e-walks} gives
that the value of $d(u)-2d_i(u)$ increases by at least 1 for every such vertex,
and never decreases. Thus $d(u) \geq 2d_i(u)$ in $G'$. 

\emph{Case 2: $i$ is not a dominant color for $u$ in $G$.}
In this case, $|X_i(u)|=d(u)-2d_i(u) \geq 0$, and we need to show that this value is nonnegative in $G'$.
By Claim~\ref{claim:e-walks}, the value of $d(u)-2d_i(u)$ decreases by at most as much
as the number of vertices in $X_i(u)$ affected by $M'$, in the sense of the term used 
in Claim~\ref{claim:M-props}. Since there are only $d(u)-2d_i(u)$ such vertices,
we see that $d(u) \geq 2d_i(u)$ also in $G'$. 

Hence we conclude that $d(u) \geq 2d_i(u)$ in $G'$ for every $u \in V(G)$ and every $i \in [k]$,
hence every vertex is balanced. This concludes the proof that $G'$ has a PC Euler trail. 
Clearly, the weight of this trail is equal to the total weight of $E(G)$ plus the sum of the weight
of the added $e$-walks, where the latter part is exactly the weight of $M'$. \\

Now assume that CPP-ECG on $G$ has a solution, a PC closed walk $Q$ in $G$, and let $G'$ be the graph obtained from $G$ by replacing every edge $e=xy$ by $q_e$ parallel edges with vertices $x$ and $y$, where $q_e$ is the number of times $Q$ traverses $e$.
Let $W = E(G') \setminus E(G)$, i.e., $W$ are the edges that are added to $G$
in order to get the PC Euler multigraph $G'$.
We will find a perfect matching in $H$ with total weight at most the sum of the weights of edges in $W$. This will complete the proof.

We initially define a set $W'$ of walks as the set of one-edge walks $xey$, where $e=xy\in W$. We will merge walks in $W'$ until we can
map the remaining walks $W'$ to a matching $M_0$ in $H$ meeting the requirements of Claim~\ref{claim:M-props}, at which point we will
be done.  {Here to merge two walks is to replace the walks $u_1e_1u_2\dots e_{\ell-1}u_{\ell}$ and $v_1f_1v_2 \dots f_{h-1}v_h$, where $u_{\ell}=v_1=u$, with the walk $u_1e_1u_2 \dots e_{\ell-1}uf_1v_2\dots f_{h-1}v_h.$} For
$u \in V(G)$ and $i \in [k]$, let $w_i(u)$
denote the number of times that $u$ is an end-vertex of a walk in $W'$ and that walk ends in $u$ with color $i$.  {Here we do not assume a fixed ``first'' and ``last" vertex, and thus the walks $u_1e_1u_2 \dots e_{l-1}u_l$ and $u_le_{l-1}u_{l-1} \dots e_1u_1$ are the same.}  Note that we will allow walks in $W'$
to start and end in the same vertex, in which case one walk may contribute to $w_i(u)$ twice.

By Claim~\ref{claim:M-props}, we need to ensure that 

\begin{enumerate}
 \item\label{it:size} $w_i(u) \leq \theta_i(u)=|X_i(u)|$ for every $u \in V(G)$ and $i \in [k]$ (so that $W'$ corresponds to a matching);
 \item\label{it:balance} $\sum_{j \in [k] \setminus \{i\}} w_j(u) \geq 2d_i(u)-d(u)$ for every vertex $u$ with a dominant color $i$ in $G$; and
 \item\label{it:parity} the parity condition is met for every vertex $u$.
\end{enumerate}

Because $G'$ has a PC closed walk   {traversing all edges}, and by Theorem~\ref{thm:kotzig}, we have that initially $W'$ satisfies the following:

\begin{enumerate}
\setcounter{enumi}{3}
 \item\label{it:Wbalance} $\sum_{j \in [k]} w_j(u) + d(u) \geq 2w_i(u) + 2d_i(u)$ for every vertex $u \in V(G)$ and integer $i$ (since $G'$ is balanced);
 \item\label{it:Wparity} $\sum_{j \in [k]}w_j(u)$ is even if and only if $d(u)$ is even i.e. the parity condition is met (since $G'$ is even, and hence $d(u)+\sum_{j \in [k]}w_j(u)$ is even).
\end{enumerate}

We note that {Condition}~\ref{it:Wbalance} implies {Condition}~\ref{it:balance} and {Condition}~\ref{it:Wparity} implies {Condition}~\ref{it:parity}. 
As long as {Condition}~\ref{it:size} is not satisfied, we will modify $W'$ by merging walks in such a way that {Condition} \ref{it:Wbalance} and {Condition}~\ref{it:Wparity} are still satisfied.
As each merging reduces the number of walks in $W'$, we must eventually stop with a set of walks $W'$ satisfying {Condition}~\ref{it:size}, {Condition}~\ref{it:balance} and {Condition}~\ref{it:parity}.

So now assume that {Condition}~\ref{it:size} is not satisfied and let $u \in V(G)$ be a vertex such that $w_i(u) > \theta_i(u)$ for some $i \in [k]$.
If $w_i(u) > d(u) - 2d_i(u)$, then we must have that $w_c(u)>0$ for some $c \neq i$, as otherwise
$2w_i(u) + 2d_i(u) > w_i(u) + d(u)  = \sum_{j \in [k]}w_j(u) + d(u)$, a contradiction to {Condition}~\ref{it:Wbalance}. 
Thus there are at least two colors $c$ with $w_c(u) > 0$.

We will choose two colors $h,j,$ with $w_h(u) >0, w_j(u) > 0, h\neq j$ ($i$ is not necessarily in $\{j,h\}$), and merge a walk ending at $u$ with color $h$ with a walk ending at $u$ with color $j$. (If this makes us merge both endvertices of the same walk,
we may simply remove the walk). It is clear that the new walk is still PC, and this operation reduces the number of walks in $W'$.
As we have reduced $w_h(u)$ and $w_j(u)$ by $1$, and the other values are unaffected, it is clear that {Condition}~\ref{it:Wparity} is still satisfied. We now show how to choose $h,j$ in such a way that {Condition}~\ref{it:Wbalance} is still satisfied.

Let us call a color $c$ \emph{at risk} if $\sum_{j \in [k]} w_j(u) + d(u) \leq  2w_c(u) + 2d_c(u) +1$ (informally, a color is ``at risk'' if removing two edges of other colors would lead that color to dominate $u$).
As $\sum_{j \in [k]} w_j(u) + d(u)$ is necessarily even and $\sum_{j \in [k]} w_j(u) + d(u) \geq 2w_c(u) + 2d_c(u)$ , we have that in fact $\sum_{j \in [k]} w_j(u) + d(u) = 2w_c(u) + 2d_c(u)$ for any at risk color $c$.
Furthermore, we note that at most two colors in $[k]$ can be at risk. Indeed, suppose that distinct colors $c_1,c_2,c_3 \in [k]$ are 
at risk. Then $2w_{c_1}(u) + 2d_{c_1}(u) + 2w_{c_2}(u) +2d_{c_2}(u) + 2w_{c_3}(u)  + 2d_{c_3}(u) = 3(\sum_{j \in[k]} w_j(u) + d(u)) > 2(\sum_{j \in [k]} w_j(u) + d(u)) \ge 2(w_{c_1}(u) + w_{c_2}(u) + w_{c_3}(u) + d_{c_1}(u) + d_{c_2}(u) + d_{c_3}(u))$, a contradiction.

Next, suppose for a contradiction that $w_c(u)=0$ for an at risk color $c$.
Then $2d_c(u) = d(u) + \sum_{j \in [k]}w_j(u)$ and so $2d_i(u) \le 2d(u)-2d_c(u) = d(u) - \sum_{j \in [k]}w_j(u)$.
Then $w_i(u) > \theta_i(u) \ge d(u) - 2d_i(u) \ge d(u) - d(u) + \sum_{k \in [k]}w_j(u) \ge w_i(u)$, a contradiction.
Thus, $w_c(u)>0$ for any at risk color $c$.

We now know that there at least two colors $c$ with $w_c(u) > 0$, there are at most $2$ at risk colors, and if color $c$ is at risk then $w_c(u) > 0$.
We can therefore select two distinct colors $h,j$ with $w_h(u)>0, w_j(u)>0$, such that any at risk color is contained in $\{h,j\}$. We now merge a walk ending with color $h$ at $u$ and a walk ending with color $j$ at $u$, as described above. This has the effect of reducing each of $w_h(u)$ and $w_j(u)$ by $1$, and leaving $w_c(u)$ unchanged for $c \in [k]\setminus\{h,j\}$. 
We now show that we still have $\sum_{j \in [k]} w_j(u) + d(u) \geq 2w_c(u) + 2d_c(u)$ for any $c \in [k]$. If $c \in \{h,j\}$, then both $\sum_{j \in [k]} w_j(u) + d(u)$ and $2w_c(u) + 2d_c(u)$ are reduced by $2$, so the condition still holds. If $c \notin \{h,j\}$, then as $c$ was not at risk we originally had $\sum_{j \in [k]} w_j(u) + d(u) \geq 2w_c(u) + 2d_c(u) + 2$. As  $\sum_{j \in [k]} w_j(u) + d(u)$ is reduced by $2$, we will still have $\sum_{j \in [k]} w_j(u) + d(u) \geq 2w_c(u) + 2d_c(u) $, as required.

We continue the above process until there is no $u \in V(G), i \in[k]$ for which {Condition}~\ref{it:size} fails.
We therefore have that {Conditions}~\ref{it:size},~\ref{it:Wbalance} and~\ref{it:parity} hold, which in turn implies {Conditions}~\ref{it:balance} and~\ref{it:parity} hold.
 Convert $W'$ to a matching $M_0$ in $H$
by adding for every walk $F$ an edge $e$ to $M_0$ such that $F$ is an $e$-walk. This is possible
since $w_i(u)\leq \theta_i(u)=|X_i(u)|$ for every $i \in [k]$, $u \in V(G)$.
The weight of $M_0$ is at most the weight of $W$, since every edge $e$ added to $M_0$ this way
has a weight corresponding to a minimum-weight $e$-walk where $F$ is just one possible $e$-walk.
By Claim~\ref{claim:M-props} we can complete $M_0$ to a perfect matching $M'$ by adding artificial
edges, which does not increase the weight. Hence $H$ admits a perfect matching whose weight
is at most the weight of $W$.

So in all cases we can find a perfect matching in $H$ with weight exactly the weight of all the (non-closed) walks in $W'$. As we have already shown that
a perfect matching in $H$ gives rise to a solution to CPP-ECG on $G$ where duplicated edges add the same weight as the 
weight of the matching, we are done.
\end{proof}

\section{Conclusion}
We considered the Chinese Postman Problem on edge-colored graphs (CPP-ECG). This problem generalizes the Chinese Postman Problem on undirected and directed graphs and the properly colored Euler trail problem on edge-colored graphs, all of which can be solved in polynomial time. We proved that CPP-ECG is still polynomial time solvable.

{
It is well-known that the number of Euler trails on digraphs can be calculated in polynomial time using the so-called  BEST theorem \cite{ST,BE}, named after de Bruijn, van Aardenne-Ehrenfest, Smith and Tutte. However, the problem is much harder on undirected graphs as it is $\#$P-complete \cite{BW}. Our simple transformation from directed walks to PC walks described in Section \ref{sec1}, shows that 
the problem of counting PC Euler trails on 2-edge-colored graphs generalizes that of counting the number of Euler trails on digraphs. Assigning each edge of an undirected graph a distinct color, shows that the problem of counting PC Euler trails on $k$-edge-colored graphs is $\#$P-complete when $k$ is unbounded. So it would be interesting to determine the complexity of the problem of counting PC Euler trails on a $k$-edge-colored graphs when $k$ is bounded, in particular when $k=2.$
}

\vspace{3mm}

\noindent{\bf Acknowledgments.} We are thankful to the referees for careful reading of the paper and several useful suggestions. Research of GG was partially supported by Royal Society Wolfson Research Merit Award. Research of BS was partially supported by China Scholarship Council.

\end{document}